\documentclass{llncs}
\usepackage{amsmath}
\usepackage{algorithm}
\usepackage{graphicx}
\usepackage[noend]{algpseudocode}
\begin{document}
\newtheorem{thm}{Theorem}
  \title{Block-Level Parallelism in Parsing Block Structured Languages}
  \author{Abhinav Jangda}
  \institute{Indian Institute of Technology (BHU), Varanasi \email{abhinav.student.apm11@iitbhu.ac.in}}
  
\maketitle
  \begin{abstract}
Software’s source code is becoming large and complex. Compilation of large base code is a time consuming process. Parallel compilation of code will help in reducing the time complexity. Parsing is one of the phases in compiler in which significant amount of time of compilation is spent. Techniques have already been developed to extract the parallelism available in parser. Current LR(k) parallel parsing techniques either face
difficulty in creating Abstract Syntax Tree or requires modification in the grammar or are specific to less expressive grammars. Most of the programming languages like C, ALGOL are block-structured, and in most language’s grammars the grammar of different blocks is independent, allowing different blocks to be parsed in parallel. We are proposing a block level parallel parser derived from Incremental Jump Shift Reduce
Parser by \cite{degano}. Block Parallelized Parser (BPP) can even work as a block parallel incremental parser. We define a set of Incremental Categories and create the partitions of a grammar based on a rule. When parser reaches the start of the block symbol it will check whether the current block is related to any incremental category. If block parallel parser find the incremental category for it, parser will parse the block in parallel. Block
parallel parser is developed for LR(1) grammar. Without making major changes in Shift Reduce (SR) LR(1) parsing algorithm, block parallel parser can create an Abstract Syntax tree easily. We believe this parser can be easily extended to LR (k) grammars and also be converted to an LALR (1) parser. We implemented BPP and SR LR(1) parsing algorithm for C Programming Language. We evaluated performance of both techniques by parsing 10 random files from Linux Kernel source. BPP showed 28\% and 52\% improvement in the case of including header files and excluding header files respectively.
\end{abstract}
\section{Introduction}
Multi core chip architectures are emerging as feasible solution to effectively utilizing the ever growing number of chip. Multi-core chip depends on success in system software technology (compiler and runtime system), in order to have thread level parallelism and utilizing on-chip concurrency. With multi-core processors, additional speedups can be achieved by the use of parallelism in data-independent tasks. There is a gradual shift towards making current algorithms and design into parallel algorithms. It is rather difficult to achieve lock free and low cache contention parallel algorithms.\\
\indent In 70’s papers appeared ideas on parallel compilation of programming languages and parallel execution of programs were expected. In those papers discussions on parallel lexical analysis, syntactic analysis and code generation were done. With VLSI applications, prominent increase in research on parallel compilation is observed.\\
\indent A compiler contains different phases: lexical analyzer, syntactic analyzer, semantic analyzer and code generator. Parsing or syntax analysis is the phase of compiler which analyses the program code according to the language. After analysis, it converts the code into another formal representation which will act as input for succeeding phases of compiler.\\
\indent Complexity of software source code is increasing. An effort to compile large code base is very time consumable. \cite{aho} describes two types of parsers: Top Down and Bottom Up Parsers. Top Down parsers have less power as compared to Bottom Up Parser. LR (k), SLR (k) and LALR (1) are types of Bottom Up Parsers. With more power Bottom Up Parsers also requires more space and more time in parsing a string as compared to Top Down parsers. Most of the compiler compilers like Yacc \cite{yacc} and Bison \cite{bison} creates LR (1) parsers and compilers like clang \cite{clang}, Mono C\# Compiler \cite{mono} etc. uses LR(1) parsers. So, it is evident that programming languages can be represented easily by LR (1) languages.\\
\indent Parsing is very time consuming phase of compiler. Parsing different files in parallel is not enough. As programming languages like C and C++ can includes different files (using \#include) in a single file which results in generation of very long file. If we can parallel the parsing phase of single file, it will give performance benefits in compiling the source code. Many significant techniques are already proposed for making parallel parsers (\cite{clarke}, \cite{cohen}, \cite{mickunas}, \cite{fischer}, \cite{ligett}). A parallel parsing for programming language is given by \cite{khanna}.\\
\indent A block is a section of code which is grouped together. In a language, a block may contain class definition, member or method declaration. Another block could be a block of statements also called compound statement. This block is usually associated with a function code or if statement or loop. Programming Languages such as C, C++, Java, Python use the concept of blocks heavily. One of most important property of parsing blocks is that they all are independent of each other i.e. each block can be parsed independently of other block. So, we could parse many blocks in a parallel fashion. In this paper, we will propose a technique to parse various blocks of the code in parallel. It can also work as a block parallel incremental parser. Our parser is termed as Block Parallelized Parser (BPP, for short).\\
\indent Our technique of parallel parsing is based on incremental parsing. An incremental parser is the one that parse only those portions of a program that have been modified. Whereas an ordinary parser must process the entire program when it is modified. An incremental parser takes only the known set of changes done in a source file and updates its internal representation of source file which may be an Abstract Syntax Tree. By building upon the previously parsed files, the incremental parser avoids the wasteful re-parsing of entire source file where most of the cod remains unchanged.\\
\indent BPP is based on the properties that an incremental parser can parse any part of a source code without the need of parsing the whole source code and different blocks in a source code can be parsed independently of other blocks. In BPP these parts are blocks in a source code. Using the property of incremental parser, BPP parse each of the blocks independently of other blocks. Each of these blocks are parsed in their own thread. It can be easily seen that BPP follows a divide and conquer approach. It divides the source into different blocks, parse each of them in parallel and at the end conquer these blocks. In our scheme the conquer step does nothing except waiting for all the BPP Threads to complete their operations.\\
\indent There have been many works on incremental parsing [Incremental Parsing References]. We choose to extend on the works of Incremental Jump Shift Reduce parser of \cite{degano}. BPP is derived from Incremental Jump Shift Reduce Parser. In \cite{degano}, authors defined Incremental Jump Shift reduce parser for SLR (1) languages only. However, we decided to extend this parser to accept LR(1) language because LR(1) languages has more power than SLR (1) and nearly all programming languages can be defined in the form of LR(1) grammars. We define the incremental categories to be a statement containing a block like class definition or function definition or if statement or for loop statement. Then, we give a notion of First Non-Terminal symbols of a Non-Terminal symbol. We used this
notion to create partitions of a grammar such that a partition
includes an incremental category and its First Non-Terminals. We
observed that this scheme gives us a very interesting property in
Incremental Jump Shift Reduce parser. We used this property to
create our Block Parallelized Parser.\\
\indent Whenever a start of the block symbol is encountered the parser will first check whether the current block is related to any incremental category and can it be parsed independently. If BPP is able to find the incremental category for it, BPP will start parsing the block in parallel. In this paper we developed this BPP for LR(1) languages but we believe it can be easily extended to LR(k) or can be easily converted to
LALR (1) or SLR (1) grammars. We also see that no major changes were done to the current LR(1) parsing algorithm and hence, it should be easy to create an Abstract Syntax Tree. This parser can also work as an incremental parallel parser which can parse different blocks in parallel. Moreover, it could be seen that there is no requirement of any Thread Synchronization to communicate between different threads of BPP each of which is parsing a block in parallel. This is because no two blocks are related in any way for the purpose of parsing.\\
\indent We compiled C\# implementation using Mono C\# Compiler 3.12.1 and executed the implementation using Mono JIT Compiler 3.12.1 on machine running Fedora 21 with Linux Kernel 3.19.3 with 6 GB RAM and Intel Core i7-3610 CPU with HyperThreading enabled. We found out that our technique showed 28\% performance improvement in the case of including header files and 52\% performance improvement in the case of excluding header files.\\
\indent The following paper is designed as follows. Section 2 shows some previous work done in parallel parsing. Section 3 and 4 provides the terminology we will use and the background required to understand our technique. In Section 5 we will extend Incremental Jump Shift Reduce parser to accept LR(1) grammars. In Section 6, we introduced the concept of First Non Terminals of a non terminal. In Section 7 we will use this concept to create partitions of the grammar. We also showed that by creating partitions using this concept we get a very interesting property. This property would be used by BPP. We have generalized this property in a theorem and also provided a proof for it. In Section 8 we presents our Block Parallelized Parser and its parsing algorithm. In Section 9 we will compare our algorithm with previous work. Section 10 shows our evaluation and results. In Section 11 and Section 12 we complete our document with conclusion and related work. 

\section{Related Work}
A lot of previous work has been done in Parallel Parsing of LR (1)
and Context Free Languages. The most recent work done by \cite{clarke} in
parallel parsing of LR(1) is an extension of an LR substring parser
for Bounded Context Languages (developed by Cormack) for
Parallel environment. \cite{cormack} provided a substring parser for Bounded
Context-LR Grammars and Simple Bounded Context-LR
Grammars. \cite{clarke} distributes the work of parsing the substrings of a
language over different processors. The work was extended to
different processors in a Balanced Binary Tree Fashion and
achieved O(log n) time complexity of parsing. But constructing a
Bounded Context LR Grammar for a programming language is
also difficult. C++ is one of the programming languages which
cannot be parsed by LR (1) parsing \cite{willink} so creating a Bounded
Context Grammar is out of question here.\\
\indent Parallel and distributed compilation schemes can be divided
into two broad categories, functional decomposition and data
decomposition. \cite{neal} and \cite{essouki} talks about distributed compilation
using a scheme based on functional decomposition. Functional
decomposition scheme divides different phases of compiler: lexer,
parser, semantic analyzer into functional component and running
each of them on separate processors like an instruction pipeline
fashion. The data decomposition scheme divide the input into
sections of equal length and parse them in parallel. BPP is data
decomposition scheme which parallel the parser by divide and
conquer approach. The data decomposition scheme was developed
by \cite{cohen}, \cite{schell}, \cite{fischer}, \cite{ligett}. These schemes are parsing LR (k) in
parallel. They divide the input into sections of equal length and
then parse them in parallel. \cite{mickunas}, \cite{ligett}, \cite{cohen} describes asynchronous
algorithms while \cite{fischer} develops a synchronous algorithm. \cite{ligett}
develops a parallel LR parser algorithm using the error recovery
algorithm of \cite{pennello}.\\
\indent \cite{neal} has developed an Incremental Parallel Compiler which
could be used in Interactive Programming Environment and he
developed an Incremental Parallel Parser also. \cite{bernardy} improves upon
the Divide and Conquer Parsing technique developed by \cite{valiant}.
They show that while the conquer step of algorithm in \cite{valiant} is
\begin{math}
  O(n^3)
\end{math} but under certain conditions it improves to \begin{math}
O(log^3 n)\end{math}
.\\
\indent \cite{khanna} describes a grammar partitioning scheme which would
help in parsing the language in parallel. In \cite{khanna} a type of
Statement Level Parallelism has been developed. The grammar is
divided into n different sub-grammars corresponding to n subsets
of the language which will be handled by each sub-compiler. For
each n sub grammars required to generate parse tables (using
parser generator) along with driver routine constitute a parser for
sub-compiler. For each sub-compiler, a requirement of modified
scanner is there which recognizes subset of the language. The
technique described in \cite{khanna} requires a lot of modification to the
Lexical Analyzer. A separate Lexical Analyzer has to be
developed for one type of language. The parser of \cite{khanna} requires
automatic tools for its implementation.\\
\indent In all the above described techniques constructing Abstract
Syntax Tree for a Block Structured Language is difficult. As
Blocks in a Block Structured Language are independent on each
other, so they can be parsed independently. Moreover this scheme
would not involve Inter-Thread Communication before a thread
completes its Parsing of Blocks. Hence, no shared memory
synchronization methods are required to coordinate between
different threads. It could be easily seen that the creation of an
Abstract Syntax Tree is also very easy. With all these required
things in mind we have developed Block Parallelized Parser for
LR (1) languages.
\section{Terminology}
We assume the notation for Context Free Grammar is represented by \begin{math}G = (N, T, S, P)\end{math} where \textit{N} is set of non-terminals, \textit{T} is set of terminals, \textit{S} is start symbol and \textit{P} is set of productions of the grammar. The language generated by \textit{G} is given as
\centerline{\begin{math}
  L (G) = \left\lbrace \omega \in T^{\ast} | S \overset{\ast}{\Rightarrow} \omega \right\rbrace 
\end{math}}
We will use the following conventions.\\
\centerline{\begin{math}
S, A, B, ... \in N
\end{math}}
\centerline{\begin{math}
a, b, ... \in T
\end{math}}
\centerline{\begin{math}
..., w, x \in T^{\ast}
\end{math}}
\centerline{\begin{math}
X, Y \in N \cup T
\end{math}}
\centerline{\begin{math}
\alpha, \beta, \gamma, ... \in (N \cup T)^{\ast}
\end{math}}
\indent Given a grammar \textit{G}, we represent its augmented grammar as \begin{math}G' = (N', T', S', P')\end{math}, where\\
\centerline{\begin{math}
  N' = N \cup \left\lbrace S' \right\rbrace
  \end{math}}
\centerline{\begin{math}
  T' = T \cup \left\lbrace \$ \right\rbrace
\end{math}}
\centerline{\begin{math}
  P' = P \cup \left\lbrace S' \rightarrow S\$ \right\rbrace
\end{math}}
Here \begin{math}S'\end{math} is called the augmented start symbol of \begin{math}G'\end{math} and \$ is the end of string marker. We denote a set of End Of String markers as EOS. \\
\indent In our paper we will represent Block Parallelized Parser as BPP, Jump Shift Reduce parser as JSR and Incremental Jump Shift Reduce parser as I\_JSR. An LR(1) item is represented as \begin{math}
  [A \rightarrow \alpha . \beta, a]
\end{math}
, where \textit{a} is the lookahead symbol. \\
\indent In a programming language, a block represents a section of code grouped together. This section of code can be a group of statements, or a group of declaration statements. For example in Java, a block corresponding to class definition contains declaration statements for fields and methods. Block corresponding to function definition can contain declaration statement for local variables or expression statements or control flow statements. A Top-Level Block is the  starting block of a program which contains definition for Classes, Functions, Import/Include statements etc. Child Blocks are contained in either Top-Level Block or another Child Block. As we proceed further, Block will be in reference to Child Block. Fig. 1, shows an example of Top-Level and Child Blocks of a Java Program.
\begin{figure}[t!]
\includegraphics[width=85mm]{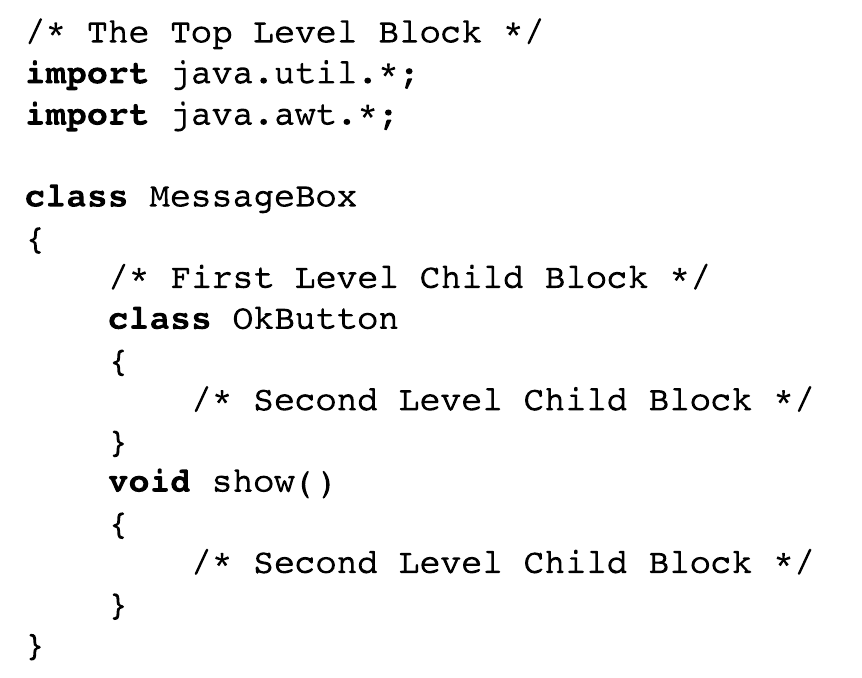}
\caption{Top-Level and Child Blocks of a Java Program.}
\end{figure}

\indent A start block symbol could be "\{" in C style languages, "begin" in Pascal style languages is represented as terminal \begin{math}s_{b}\end{math}. An end block symbol which could be "\}" in C style languages or "end" in Pascal style languages is represented as \begin{math}e_{b}\end{math}.
\section{Background}
We now survey LR (1) parsers and their generation algorithm as
given by \cite{aho}. LR (1) parsers are table driven Shift Reduce parsers.
In LR (1), “L” denotes left-to-right scanning of input symbols and “R” denotes constructing right most derivation in reverse. Some extra information is indicated with each item: the set of possible terminals which could follow the item’s LHS. This set of items is called the lookahead set for the item. Here, “1” signifies that number of lookahead symbols required are 1.\\
\indent LR (1) parser consists of an input, an output, a stack, a driver
routine. A driver routine runs its parsing algorithm which interacts
with two tables ACTION and GOTO. Any entry in ACTION and GOTO tables are indexed by a symbol which belongs to \begin{math}
  N \cup T'
\end{math} and the current state. An entry in both the tables can be any one of the following:
\begin{itemize}
  \item If ACTION [j, a] = \begin{math}<\end{math}S, q\begin{math}>\end{math} then a Shift Action must be taken.
  \item If ACTION [j, a] = \begin{math}<\end{math}R, \begin{math}
    A \rightarrow \alpha>
  \end{math} then reduce symbols to a production.
  \item If ACTION [j, a] = Accept then grammar is accepted.
  \item If ACTION [j, a] = error then Error has occurred.
  \item If GOTO [j,A] = q then go to state q.
\end{itemize}
\indent An LR (1) item is of the form \begin{math}
  [A \rightarrow \alpha . \beta,~a]
\end{math}, where \textit{a} is a lookahead symbol. Construction of LR (1) items requires two procedures CLOSURE and GOTO. CLOSURE takes a set of items as its argument. GOTO takes a set of items and a symbol as
arguments. Both are defined as follows:\\
\centerline{\begin{math}
  CLOSURE (I) = I \cup \{[B \rightarrow .\gamma,~b]~|~[A \rightarrow \alpha. B\beta,~a] \in I\end{math}}
  \centerline{\begin{math}and~B \rightarrow \gamma \in P'~\forall~b\in FIRST (\beta a)\} 
\end{math}}
\centerline{\begin{math}
  GOTO (I,~X) = CLOSURE (\{[A \rightarrow \alpha X. \beta,~a]~|~[A \rightarrow \alpha . X\beta, a] \in I\})
\end{math}}
Collection of set of LR (1) items is created using CLOSURE and GOTO functions. Items function creates the collection of set of LR (1) items.\\
\centerline{\begin{math}
  items (G') = CLOSURE (\{[S' \rightarrow .S,~\$]\}) \cup \{GOTO (I,X)~|~I \in C~and~X \in P'\}
\end{math}}
\indent ACTION and GOTO tables are created using this collection. Following is the procedure to create these tables:
\begin{enumerate}
  \item Create collection of set of LR(1) items. Let this collection be   \begin{math}
    C' = \{ I_0,~I_1,~I_2,...,~I_n\}
  \end{math}
  \item Let \textit{i} be the state of a parser constructed from \begin{math}
    I_i
  \end{math}. Entries in ACTION table are computed as follows:
  \begin{enumerate}
    \item ACTION [\textit{i}, \textit{a}] = shift \textit{j}, if \begin{math}
      [A \rightarrow \alpha .a \beta,~b] \in I_i
    \end{math} and \begin{math}
      GOTO (I_i, a) = I_j
    \end{math}
    \item ACTION [\textit{i}, \textit{a}] = reduce \begin{math}
      A \rightarrow \alpha.
    \end{math}, if \begin{math}
      [A \rightarrow \alpha., a] \in I_i 
    \end{math} and \begin{math}
      A \neq S'
    \end{math}
    \item ACTION [\textit{i}, \$] = accept, if \begin{math}
      [S' \rightarrow S., \$] \in I_i
    \end{math}
    \item GOTO [\textit{i}, \textit{A}] = \textit{j}, if GOTO (\begin{math}
      I_i
    \end{math}, A) = \begin{math}
      I_j
    \end{math}
  \end{enumerate}
  \item All other entries not defined by (b) and (c) are set to error.
  \item The Initial State is the one containing the item \begin{math}
    [S' \rightarrow .S,~\$]
  \end{math}.
\end{enumerate}
\indent Most of the programming languages could be constructed from LR (1) grammar. Hence, LR (1) is the most widely used parser. Many parser generators like YACC and GNU Bison generates an LR (1) parser.\\
\indent A Jump Shift Reduce \cite{degano} (JSR in short) parser is an extension
of LR (1) parser. LR (1) parser generates ACTION and GOTO table for the augmented grammar \begin{math}
  G'
\end{math}. JSR parser first partition the grammar \begin{math}
  G'
\end{math} into several sub grammars and creates parsing sub-
table of every sub grammar. Hence, the ACTION and GOTO tables of LR (1) are split into several ACTION and GOTO tables in JSR parser. JSR parser is equivalent to the LR (1) parser that is it will only accept languages
generated by LR (1) grammar \cite{degano}. \\
\indent Let \begin{math}
  G' = (N', T', S', P')
\end{math} be the augmented grammar of grammar \begin{math}
  G = (N, T, S, P)
\end{math}. Let us partition the grammar G on the basis of
Non Terminals. Let \begin{math}
  G^i
\end{math} denotes a partition of the grammar \textit{G}, such that we have\\
\centerline{\begin{math}
  G^i = (M^i, T^i, S^i, P^i) 
\end{math}}\\
where,
\centerline{\begin{math}
  N^i \subseteq N~such~that~N^i \cap N^j = \phi~|~ i,~j = 1,..., n, i \neq j 
\end{math}}
\centerline{\begin{math}
  P^i = \left \lbrace A \rightarrow \alpha \in P~|~A \in N^i \right \rbrace \forall i = 1,...,n  
\end{math}}
\centerline{\begin{math}
  M^i = N^i \cup \left \lbrace B \in N~|~ \exists A \rightarrow \alpha B \beta \in P^i \right \rbrace
\end{math}}
\centerline{\begin{math}
  T^i = \left \lbrace a \in T^i~|~\exists A \rightarrow \alpha a \beta \in P^i \right \rbrace
\end{math}}
\centerline{\begin{math}
  S^i \in N^i
\end{math}}\\
Therefore, we have\\
\centerline{\begin{math}
  (\cup N^i,~\cup T^i,~S',~\cup P^i) = G'
\end{math}}
\indent For every subgrammar \begin{math}
  G^i
\end{math}, a parsing subtable named Tab(\begin{math}
  S^i
\end{math}) is built. Each subtable contains ACTION and GOTO subtables which are represented by Tab(\begin{math}
  S^i
\end{math}).ACTION and Tab(\begin{math}
  S^i
\end{math}).GOTO. In addition to the Shift, Reduce and Accept action there is an additional Jump action. Jump action is associated with a sub-table. Whenever a Jump action is encountered then the parsing algorithm Jumps to a sub-table and parse the partition related to that sub-table.\\
\indent We will now investigate few points about the Incremental Jump Shift Reduce Parser \cite{degano}. Incremental Jump Shift Reduce Parser (I\_JSR) \cite{degano} is based upon the JSR parser \cite{degano}. A set of Incremental Categories will be defined which could be incrementally parsed by I\_JSR parser. Given a grammar \begin{math}
  G = (N, T, S, P)
\end{math} a set of Incremental Categories has to be defined \begin{math}
  IC = \{C_j~|~C_j\in N,~j=1,2,...,n\}
\end{math} and the Incremental Language of G is \\
\centerline{\begin{math}
  L^\ast(G) = \cup L^\ast(A),
\end{math}} where \begin{math}
  A \in IC \cup \{S\}
\end{math} and \begin{math}
  L^\ast(A)=\{\alpha \in (T\cup IC)^*~|~A \overset{\ast}{\Rightarrow} \alpha\}
\end{math}\\
\indent For every Incremental Category \textit{A}, add a production \begin{math}
  A \rightarrow A\#_A
\end{math}, where \begin{math}
  \#_A
\end{math} is an end-of-string for \textit{A}. For a given grammar \begin{math}
  G = (N, T, S, P)
\end{math} with a set of Incremental Categories \begin{math}
  IC = \{C_j~|~C_j\in N,~j = 1,2,...n\}
\end{math} an Incremental Grammar is defined as\\
\centerline{\begin{math}
  G^\ast = (N, T\cup EOS, S, P \cup P_IC),
\end{math}} where,
\centerline{
EOS is the set of End of String markers = \begin{math}
  \{\#_j~|~j=1,2,..n\}
\end{math}}\\
\centerline{\begin{math}
  P_{IC} = \{C_j \rightarrow C_j \#_j~|~C_j \in IC,~\#_j \in EOS\}
\end{math}}\\
\indent A major change by this extension is that now the strings may contain incremental symbols which are also non-terminal symbols. The difference between ACTION and GOTO tables disappears, as an incremental category can also occur in the input string and can be shifted on the stack. Hence, we would have only ACTION table and no need for GOTO table. As we have also introduced EOS set, the ACTION table can
now be indexed with symbols belonging to \begin{math}
  N \cup T' \cup EOS
\end{math}. Every Incremental Category will have its own start state and accept action. We will represent the subtable as Tab(\begin{math}S^t\end{math}). Entries of table will be as follows:
\begin{itemize}
  \item If Tab(\begin{math}S^t\end{math}) [\textit{j}, \textit{X}] = \begin{math}<\end{math}S, q\begin{math}>\end{math} then a Shift Action must be taken.
  \item If Tab(\begin{math}S^t\end{math}) [\textit{j}, \textit{X}] = \begin{math}<\end{math}R, \begin{math}
    A \rightarrow \alpha .>
  \end{math} then a Reduce Action must be taken.
  \item If Tab(\begin{math}S^t\end{math}) [\textit{j}, \textit{\$}] = Accept then input is accepted.
  \item If Tab(\begin{math}S^t\end{math}) [\textit{j}, \begin{math}
    \#_i
  \end{math}] = Accept then input for Incremental Category \begin{math}
    C_i
  \end{math} will be accepted.
  \item If Tab(\begin{math}S^t\end{math}) [\textit{j}, \textit{a}] = \begin{math}<\end{math}J, K\begin{math}>\end{math} then jump to a subtable Tab(\begin{math}
    S^k
  \end{math}).
  \item If Tab(\begin{math}S^t\end{math}) [\textit{j}, \textit{a}] = error then error occurred.
\end{itemize}
\section{Extending I\_JSR Parser to accept languages generated by LR (1) grammars}
As JSR Generation Algorithm was already developed for LR (0) items and Incremental JSR Generation Algorithm was developed for SLR (0) items \cite{degano}. In this section, we will first extend the JSR Parser Generation Algorithm to accept LR (1) Grammar and then we will extend I\_JSR parser to accept LR(1) Grammar.\\
\indent Generation of subtables first requires the generation of canonical collection of sets of augmented items. An augmented item is a triplet represented as \begin{math}<\end{math}i, FF, TF\begin{math}>\end{math}, where i is an LR item, FF called From-Field and TF called To-Field are the names of parsing sub-tables. From-Field represents the sub-table which contains the last action performed by parser and To-Field represent the sub-table which contains the next action to be performed. Although we focus only LR (1) items but the procedure for LR (k) items is very similar.\\
\indent To-Field of an augmented item of a state is determined using
TO function. Let us define the TO function.
\begin{algorithm}
\begin{algorithmic}[1]
\Procedure {TO}{$I_j$}
\State $I_j'' = \phi$
\ForAll {item \begin{math} i_k \end{math} in \begin{math} I_j \end{math}}
\If {$i_k~is~[A \rightarrow \alpha . a \beta,~b]$}
\State $add~ < i_k, CHOOSE_NEXT(I_j, a)>~to~I_j''$
\ElsIf {$i_k~is~[A \rightarrow \alpha . B \beta,~b]$}
\State $add~ < i_k, S^i>~to~I_j'', where~B\in N^i$
\ElsIf {$i_k~is~[A \rightarrow \beta .,~b]$}
\State $add~ < i_k, S^i>~to~I_j'', where~A\in N^i$
\EndIf
\EndFor
\Return {$I_j''$}
\EndProcedure
\end{algorithmic}
\end{algorithm}
Function TO calls Function CHOOSE\_NEXT. This function selects a parsing table out of those in which the parsing of the remaining string could continue. Let \begin{math}
  <\cdot
\end{math} be a total ordering relation over the set \begin{math}
  ST = \{S^p~|~p=1,2,...,n\}
\end{math} of the names of parsing sub-tables, such that \begin{math}
  S^i <\cdot S^{i+1},~i=1,2,...n-1
\end{math}\\
\begin{algorithm}
\begin{algorithmic}[1]
\Procedure {CHOOSE\_NEXT}{$I_j,~a$}
\State $let~ST'~be~a~set~of~parsing~subtables$
\ForAll {$item~[H \rightarrow \alpha. a \beta,~b]~in~I_j$}
\State $add~S_h~to~ST'~such~that~H\in N_h$
\EndFor
\Return {$min <\cdot ST'$}
\EndProcedure
\end{algorithmic}
\end{algorithm}
\indent From-Field of an augmented item is enriched using FROM function. FROM takes two arguments \begin{math}
  I_t''
\end{math}, which is a set of items enriched with To-Field and \begin{math}
  I_j
\end{math}, whose items we have to enrich with From-Field.\\
\begin{algorithm}
\begin{algorithmic}[1]
\Procedure{FROM}{$I_t'',~I_j$}
\State $I_j' = \phi$
\ForAll {$item i_k \in I_j$}
\If {$i_k~is~[S'\rightarrow .S\$, b]$}
\State $add < i_k,~S^1> to I_j$
\ElsIf {$i_k~is~[A\rightarrow \alpha X.\beta,~b]$}
\State $add <i_k,~TF>~to~I_j', where~<[A \rightarrow \alpha .X\beta,~b],~TF> \in I_t''$
\ElsIf {$i_k~is~[A\rightarrow .\beta,~b]$}
\State $add <i_k,~FF>~to~I_j', where~<[B \rightarrow \alpha .A\beta,~b],~FF> \in I_j'$
\EndIf
\EndFor
\Return $I_j'$
\EndProcedure
\end{algorithmic}
\end{algorithm}
\indent STATES procedure is used to generate the collection of set of JSR items. STATES algorithm first generates the collection of sets of LR(1) items using ITEMS procedure which was discussed previously. Afterwards, it calls TO and FROM functions to generate set of augmented items from the corresponding set of LR(1) items. \\
\begin{algorithm}
\begin{algorithmic}[1]
\Procedure {STATES} {$G'$}
\State $C' = items(G')$
\State $I^A = \phi$
\ForAll {$I_j \in C'$}
\State $I_j'=\{<i_k, FF_k>~|~i_k\in I_j\} = FROM(I_t'',~I_j),~where~I_j=GOTO(I_t,~X)~and~I_t=\phi~if~j=0$
\State $I_j''=\{<i_k, TF_k>~|~i_k\in I_j\}=TO(I_j)$
\State $I_j^A=\{<i_k, FF_k,TF_k>~|~<i_k,FF_k>\in I_j'~and~<i_k, TF_k> \in I_j''\}$
\State $I^A = I^A \cup \{I_j^A\}$
\EndFor
\Return $I^A$
\EndProcedure
\end{algorithmic}
\end{algorithm}
\indent We will now extend I\_JSR parser to accept LR (1) languages. This extended parser is based on the previous JSR parsing algorithm. The FIRST function used to compute the set of First symbols related to a non-terminal has to be modified to include the incremental categories also. The reason being an incremental category can also occur in the input string, can be shifted on the stack while parsing and can reduce to a
production. Moreover, FIRST should now also include the EOS markers. Hence, the new FIRST becomes\\
\centerline{\begin{math}
  FIRST (A) = \{a~|~A\overset{\ast}{\Rightarrow} a\beta,~where~a \in T' \cup EOS \cup IC \} 
\end{math}}
\indent For an incremental category \textit{A}, there will be a set of items containing item \begin{math}  [A \rightarrow .A\#_A, \#_A] \end{math} and items 
\begin{math}  [B \rightarrow \alpha.A\beta \#_A] \end{math}. Then
the state corresponding to this set of items will be the start state of
the incremental grammar corresponding to A. Correspondingly, there will a single set of items that contains the item \begin{math} [A \rightarrow A.\#_A, \#_A] \end{math}. The state belonging to this set of item is the accepting state of the incremental grammar corresponding to A. The I\_JSR parser has initial and final states for every incremental category.\\
\indent Now, we can extend the I\_JSR parser for accepting languages generated by LR(1) grammars. The procedure I\_TABS given below is used to construct the I\_JSR parsing table.
\begin{algorithm}
\begin{algorithmic}[1]
\Procedure {I\_TABS}{$G'$}
\State $C' = items(G')$
\ForAll {$I_j^A \in C'$}
\ForAll {$<i_h,H,K>\in I_j^A$}
\If {$i_h~is~of~the~form~[A \rightarrow \alpha .X\beta, b]$}
\State $Tab(S^K)[j, X] = <S, q>~where~GOTO (I_j, X) = I_q$
\If {$H\neq K$}
\State $Tab (S^H)[j,X]=<J, K>$
\EndIf
\ElsIf {$i_h~is~of~the~form~[A \rightarrow \alpha ., X]$}
\State $Tab(S^K)[j, X] = <R, A \rightarrow \alpha>$
\If {$H\neq K$}
\State $Tab(S^H)[j, X] = <J, K>$
\EndIf
\ElsIf {$i_h~is~of~the~form~[S' \rightarrow S. \$, \$] $}
\State $Tab(S^{S'})[j, \$]=accept$
\ElsIf{$i_h~is~of~the~form~[A \rightarrow A.\#_A, \#_A]$}
\State $Tab(S^K)[j, \#_A] = accept$
\EndIf
\EndFor
\EndFor
\EndProcedure
\end{algorithmic}
\end{algorithm}
\section{First Non Terminals}
In this section we will define the concept of First Non Terminals.\\
\indent We define a set of First Non Terminals for a non terminal \textit{A} as a set of non-terminals that appear at the beginning of any sentential form derived from \textit{A} i.e. a set of non terminals \textit{B} such that there exists a derivation of the form \begin{math}
  A \overset{\ast}{\Rightarrow} B\beta
\end{math}. \begin{math}
  FIRSTNT (A)
\end{math}
represents the set of First Non Terminals for \textit{A} and can be represented in set notations as:\\
\centerline{\begin{math}
  FIRSTNT (A) = \displaystyle\bigcup_{B} \left\lbrace B~|~A \overset{\ast}{\Rightarrow} B\beta \right\rbrace 
\end{math}}
To compute \begin{math}
  FIRSTNT (A)
\end{math} for any non-terminal A, apply the following rules until no more terminals can be added to the \begin{math}
  FIRSTNT (A)
\end{math} set.
\begin{enumerate}
  \item If \textit{A} is a terminal, then \begin{math}
    FIRSTNT (A) = \phi
  \end{math}
  \item If \textit{A} is a non-terminal and \begin{math}
    A \rightarrow B_1 B_2 ... B_k 
  \end{math}
  is a production for some \begin{math}k \leq 1\end{math}, then place \begin{math}
    B_i 
  \end{math} and everything in \begin{math}
    FIRSTNT (B_i) 
  \end{math} in \begin{math}
    FIRSTNT (A) 
  \end{math} if \begin{math}
    B_1, B_2, ... B_{i-1} \overset{\ast}{\Rightarrow} \epsilon 
  \end{math}
  \item If \begin{math} A \rightarrow \epsilon \end{math} is a production, then \begin{math}
    FIRSTNT (A) = \phi
  \end{math}
\end{enumerate}
EXAMPLE 1: If we have following productions:
  \begin{itemize}
  \item[] \begin{math}
    S \rightarrow D A B
  \end{math}
  \item[] \begin{math}
    S \rightarrow C
  \end{math}
  \item[] \begin{math}
    A \rightarrow a B
  \end{math}
  \item[] \begin{math}
    B \rightarrow b
  \end{math}
  \item[] \begin{math}
    C \rightarrow c
  \end{math}
  \item[] \begin{math}
    D \rightarrow d
  \end{math}
  \end{itemize}
Then find \begin{math}
  FIRSTNT (S)
\end{math}?

SOLUTION 1:\\
Due to first two productions of S we have,\\
\centerline{\begin{math}
  FIRSTNT (S) = FIRSTNT (A) \bigcup FIRSTNT (C) \bigcup \left\lbrace A, C \right\rbrace
\end{math}}
\centerline{\begin{math}
  FIRSTNT (A) = FIRSTNT (a) = \phi
\end{math}}
\centerline{\begin{math}
  FIRSTNT (C) = FIRSTNT (c) = \phi
\end{math}}
\section{Using First Non Terminals to create partitions}
In this section we will use the concept of First Non Terminals to create partitions of grammar. I\_JSR Parser will use these partitions to develop its tables. We will see that this kind of partitioning leads to a very interesting property in I\_JSR Parsing algorithm. We will generalize this property in a theorem and will also prove it.\\
\indent We will partition the grammar in such way that:
\begin{itemize}
  \item Every Incremental Category will have its own partition.
  \item The partition of Incremental Category will contain First Non Terminals of that incremental category also.
  \item Intersection of set of First Non-Terminals of any two incremental categories must be empty.
  \item All other remaining non-terminals including the augmented and start symbol are included in the first partition.
\end{itemize}
\indent Given a grammar \begin{math}
  G^\ast = (N,~T~\cup~EOS,~S,~P~\cup~P_{IC}) 
\end{math} with a set of Incremental Category, \begin{math}
  {IC = \left\lbrace C_t~|~C_t \in N \right\rbrace}\end{math} we define partitions of non-terminals as \begin{math}
  N^2, N^3, ..., N^n
\end{math} such that:\\
\centerline{\begin{math}
  N^t = \left\lbrace C_t \right\rbrace \cup FIRSTNT (C_t)~and~C_t \neq S
\end{math}}
and\\
\centerline{\begin{math}
  FIRSTNT (C_t) \cap FIRSTNT(C_s) = \phi,~where~t,~s~=~1,~2,~...,n~and~t~\neq~s
\end{math}}
And first partition,\\
\centerline{\begin{math}
  N^1 = N - \displaystyle\bigcup_t N^t
\end{math}}

Example 2: Partition the grammar given in Example 1 using first non terminals as given above and create I\_JSR parsing table. Then parse the string "\textit{dabb}" and parse string "\textit{ab}" incrementally for incremental category \textit{A}.

Solution 2: First we have to create an augmented grammar of the given grammar by adding production \begin{math}
  S' \rightarrow S\$
\end{math}. Next we can create two partitions of grammar as given by the following partitions of non terminals.\\
\centerline{\begin{math}
  N^1 = \left\lbrace S',~S,~B,~C,~D\right\rbrace
\end{math}}
\centerline{\begin{math}
  N^2 = \left\lbrace A \right\rbrace 
\end{math}}
\indent As the incremental category we want is only \textit{A}. Hence, there will be two partitions, \begin{math}
  N^2
\end{math} containing \textit{A} and its first non terminals and \begin{math}
  N^1
\end{math} will contain the remaining non terminals. Moreover, we would also have to add an EOS marker for \textit{A}, let us say it is \begin{math}\#_1\end{math}. We can generate the I\_JSR table using \textit{I\_TABS} procedure described in Section 5.\\
\indent Let us parse the string "\textit{dabb}". Table 1 shows the series of actions taken while parsing "\textit{dabb}". In this case the start state will be the start state of table \begin{math}
  Tab (S')
\end{math} i.e. 0.
Table 2 shows the series of actions taken when "\textit{ab}" is parsed incrementally with the incremental category \textit{A}. In this case the start state will be the start state of table Tab (A) i.e. 2.\\
\begin{table}[ht]
\begin{minipage}[b]{0.45\linewidth}\centering
\caption{Parsing of "\textit{dabb}"}
\begin{tabular}{|l|r|l|}
\hline
STACK      & INPUT & ACTION                        \\ \hline
0        & dabb\$  & Shift 5                        \\ \hline
0d5      & abb\$   & Reduce \begin{math}D \rightarrow d \end{math}     \\ \hline
0D2      & abb\$   & Jump A                     \\ \hline
0D2      & abb\$   & Shift 7                    \\ \hline
0D2a7    & bb\$    & Jump S'                    \\ \hline
0D2a7    & bb\$    & Shift 11                   \\ \hline
0D2a7b11 & b\$     & Reduce \begin{math}B \rightarrow b \end{math}     \\ \hline
0D2a7B10 & b\$     & Jump A                     \\ \hline
0D2a7B10 & b\$     & Reduce \begin{math}A \rightarrow a B\end{math}  \\ \hline
0D2A6    & b\$   & Shift 9                      \\ \hline
0D2A6b9  & \$     & Reduce \begin{math}B \rightarrow b\end{math}     \\ \hline
0D2A6B8  & \$     & Reduce \begin{math}S \rightarrow DAB\end{math}   \\ \hline
0S1      & \$     & Accept                     \\ \hline
\end{tabular}
\end{minipage}
\hspace{0.5cm}
\begin{minipage}[b]{0.45\linewidth}
\caption{Parsing of "\textit{ab}"}
\label{my-label}
\begin{tabular}{|l|r|l|}
\hline
STACK      & INPUT & ACTION                  \\ \hline
2          & ab\begin{math}\#_1\end{math}  & Shift 7                     \\ \hline
2a7        & b\begin{math}\#_1\end{math}  & Jump S'                      \\ \hline
2a7        & b\begin{math}\#_1\end{math}  & Shift 11                      \\ \hline
2a7b11     & \begin{math}\#_1\end{math}   & Reduce \begin{math} B \rightarrow b\end{math}     \\ \hline
2a7B10     &\begin{math}\#_1\end{math}    & Jump A                          \\ \hline
2a7B10     & \begin{math}\#_1\end{math}   & Reduce \begin{math} A \rightarrow a B \end{math}     \\ \hline
2A6        & \begin{math}\#_1\end{math}    & Accept                            \\ \hline
\end{tabular}
\end{minipage}
\end{table}

\indent In Table 1, when stack state reaches 0D2 there is a Jump to the Table Tab (A). From this point until when stack state changes to 0D2A6, the actions taken are same as the actions of Table 2 and in the same table i.e. Tab (A). Moreover, in between these Stack states in Table 1 "\textit{ab}" is parsed to \textit{A}.\\
\indent We can generalize this example in the sense that same series of actions are taken when parsing a string and when parsing its substring incrementally for its incremental category. In the current example all the same actions happens in the same table because we created the partitions in such a way that all the first non terminals are in that partition. If the partitions were not created in the way described, it could have happened that these actions would happen in the different sub tables.\\
\indent This technique is utilized by our BPP and it is generalized and proved in the theorem below.
\begin{thm}
Given an Incremental Grammar \begin{math}
  G^\ast = (N, T, P', S')
\end{math} with set of incremental categories \begin{math}
  IC = \left \lbrace C_t~|~C_t \in N \right \rbrace 
\end{math} such that the non terminal partition, \begin{math}
  N^t
\end{math} related to incremental category \begin{math}
  C_t
\end{math}contains only \begin{math}
  C_t
\end{math} and \begin{math}
  FIRSTNT (C_t)
\end{math}. For an incremental category \begin{math}
  C_t
\end{math} and any terminal \begin{math}
  b \in FIRST (C_t)
\end{math}, if \begin{math}
  C_t \overset{\ast}{\Rightarrow} b\gamma
\end{math}
and during parsing of the word \begin{math}
  w="\mu b\gamma \delta"
\end{math} the parser reaches at a state \textit{q} in the subtable of \begin{math}
  C_t
\end{math} after performing the shift action on \textit{b} then during the incremental parsing of the word \begin{math}
  w_t = "b\gamma"
\end{math} for the incremental category \begin{math}
  C_t
\end{math} the parser will also reach the state \textit{q} after performing the shift action on \textit{b} in the sub table of \begin{math}
  C_t
\end{math}.
\end{thm}
\begin{proof}
\textbf{Outline}: We will divide the proof in 4 cases. For each case, we will first find two sets of JSR items reached after performing shift on \textit{b} one during the parsing of word \begin{math}
  w
\end{math} and other during the incremental parsing of word \begin{math}
  w_t
\end{math} related to incremental category \begin{math}
  C_t
\end{math}. We will then show that both of these sets contains same JSR items which implies the above theorem.\\
It is given that \begin{math}
  C_t
\end{math} is an incremental category and \begin{math}
  C_t \neq S
\end{math} and let \begin{math}
  N^t
\end{math} be the partition containing only \begin{math}
  C_t
\end{math} and \begin{math}
  FIRSTNT (C_t)
\end{math}. Let \begin{math}
  S^t
\end{math} be the sub-table related to \begin{math}
  C_t
\end{math}. For any non-terminal \begin{math}
  A \in FIRSTNT (C_t)
\end{math} and \begin{math}
  A \rightarrow b\beta
\end{math}, we must have \begin{math}
  A \in N^t
\end{math}.\\
\indent During incremental parsing of the word \begin{math}
  w_t = b\gamma
\end{math} for the incremental category \begin{math}
  C_t
\end{math}, the state before performing actions related to \textit{b} will be the start state of the sub table \begin{math}
  Tab (C_t)
\end{math}. Let that start state be \textit{m}.\\
We will have four cases on the basis of whether the grammar has productions of the form, \begin{math}
  B \rightarrow cXC_t\beta
\end{math} and \begin{math}
  B \rightarrow C_t \beta
\end{math}\\
\textbf{Case 1}: If \begin{math}
  B \rightarrow cXC_t\beta \in P'
\end{math}
and
\begin{math}
  B \rightarrow C_t\beta \in P'
\end{math}\\
\indent Let the set of LR(1) items related to the start state of \begin{math}
  Tab (C_t)
\end{math} i.e. state \textit{m} be \begin{math}
  I_m
\end{math} and \textit{X} lies in some partition other than \begin{math}
  N^t
\end{math} i.e. \begin{math}
  X \in N^h
\end{math} and \begin{math}
  t \neq h
\end{math}.

As noted by \cite{degano} the start state of \begin{math}
  Tab (C_t)
\end{math} must contain the item \begin{math}
  [C_t \rightarrow .C_t \#_t, d]
\end{math}, where \textit{d} is a lookahead symbol. So we have
\begin{math}
  [C_t \rightarrow .C_t \#_t, d] \in I_m
\end{math}. It is evident that the item \begin{math}
  [C_t \rightarrow .C_t \#_t, d]
\end{math} should be result of a closure of another item. The only such item we can see is \begin{math}
  [B \rightarrow cX.C_t\beta, j]
\end{math}, where \textit{j} is some lookahead symbol. So, we must have \begin{math}
  [B \rightarrow cX.C_t\beta, j] \in I_m
\end{math}. As discussed in Section 2, to get a set of LR(1) items we have to apply CLOSURE (\begin{math}
  [B \rightarrow cX.C_t\beta, j]
\end{math}). Hence, we have\\
\centerline{\begin{math}
  [B \rightarrow c.XC_t\beta, j] \in I_m
\end{math}}\\
\centerline{\begin{math}
  [C_t \rightarrow .C_t\#_t, d] \in I_m
\end{math}}\\
\centerline{\begin{math}
  [C_t \rightarrow .A\gamma \beta, e] \in I_m
\end{math}}\\
\centerline{\begin{math}
  [A \rightarrow .b\beta, f] \in I_m
\end{math}}\\
where \textit{e} and \textit{f} are some lookahead symbols.\\
\indent Let \begin{math}
  I_o
\end{math} be the set of LR(1) items such that \begin{math}
  I_m = GOTO (I_o, X)
\end{math}. So, \begin{math}
  [B \rightarrow cX.C_t\beta, j] \in I_o
\end{math}. Let \begin{math}
  I_o''
\end{math} be a set of JSR items enriched with TO fields corresponding to all LR(1) items in \begin{math}
  I_o
\end{math}. After applying TO procedure over \begin{math}
  I_o
\end{math} we would get the TO field for item \begin{math}
  [B \rightarrow cX.C_t\beta, j] 
\end{math} as \begin{math}
  S^h
\end{math} because \begin{math}
  X \in N^h
\end{math}.\\
\indent To get TO fields for all JSR items corresponding to LR (1) items of the set \begin{math}
  I_m
\end{math} we have to apply TO procedure over \begin{math}
  I_m
\end{math}. We could see that the TO fields for items \begin{math}
  [B \rightarrow cX.C_t\beta, j]
\end{math}, \begin{math}
  [C_t \rightarrow .C_t\#_t, d]
\end{math},
\begin{math}
  [C_t \rightarrow .A\gamma, e]
\end{math},
\begin{math}
  [A \rightarrow .b\beta, f]
\end{math}
will be \begin{math}
  S^t
\end{math}. \\
\indent To enrich JSR items for all the LR(1) items in \begin{math}
  I_m
\end{math} with FROM field, we would apply FROM procedure as \begin{math}
  FROM (I_o'', I_m)
\end{math}. Now we could see that the FROM field of JSR item of LR(1) item \begin{math}
  [C_t \rightarrow .A\gamma, e]
\end{math} will be equal to the FROM field of \begin{math}
  [B \rightarrow cX.C_t\beta, j]
\end{math} which in turn is equal to the TO field of \begin{math}
  [B \rightarrow c.XC_t\beta, j]
\end{math}, which is \begin{math}
  S^h
\end{math}. So, the set of JSR items (\begin{math}
  I_m^A
\end{math}) corresponding to \begin{math}
  I_m
\end{math} contains the following items:\\
\centerline{\begin{math}
  [B \rightarrow cX.C_t\beta, j, S^h, S^t] \in I_m^A
\end{math}}\\
\centerline{\begin{math}
  [C_t \rightarrow .C_t\#_t, d, S^h, S^t] \in I_m^A
\end{math}}\\
\centerline{\begin{math}
  [C_t \rightarrow .A\gamma, e, S^h, S^t] \in I_m^A
\end{math}}\\
\centerline{\begin{math}
  [A \rightarrow .b\beta, f, S^h, S^t] \in I_m^A
\end{math}}\\
Let after performing shift operation in the state \textit{m} over the symbol \textit{b}, parser reaches state \textit{n}. So, we must have\\
\centerline{\begin{math}
  I_n = GOTO (I_m, a) = CLOSURE (\{[A \rightarrow b.\beta, g]~|~\forall~g \in FIRST (\beta f)\})
\end{math}}\\
Moreover, the TO and FROM fields of all items in the above state will be \begin{math}
  S^t
\end{math}. We have obtained the set of JSR items reached after performing shift over the symbol \textit{b} during incremental parsing of word \begin{math}
  w_t
\end{math} for incremental category \begin{math}
  C_t
\end{math}. Also, the subtable at this state is \begin{math}
  Tab (S^t)
\end{math}.\\
\indent We will now obtain the set of JSR items reached after performing shift over the symbol \textit{b} during the parsing of word \begin{math}
  w
\end{math} for incremental category \begin{math}
  C_t
\end{math}. Let us consider the derivation \begin{math}
  S \overset{\ast}{\Rightarrow} \eta B \overset{\ast}{\Rightarrow} \mu b \gamma \delta
\end{math} such that, \textit{B} doesn't derive \begin{math}
  \eta
\end{math}. In the above derivation only one production out of \begin{math}
  B \rightarrow cXC_t\beta
\end{math} and \begin{math}
  B \rightarrow C_t \beta
\end{math} will be used. Hence, we have two cases of the basis of which production is used.\\
\\
\textit{Case 1.1}: If production \begin{math}
  B \rightarrow cXC_t\beta
\end{math} is used, then the set of items related to state reached just after performing shift on \textit{c} in the above derivation say \begin{math}
  I_x
\end{math} will contain the LR(1) item \begin{math}
  [B \rightarrow c.XC_t\beta, j]
\end{math} besides other items of the form \begin{math}
  [X \rightarrow .\delta, k]
\end{math}, where \begin{math}
  k \in FIRST (C_t\beta c)
\end{math}. After applying TO procedure on \begin{math}
  I_x
\end{math}, we should get the TO field for LR(1) item \begin{math}
  [B \rightarrow c.XC_t\beta, j]
\end{math} as \begin{math}
  S^h
\end{math}. Let us represent these set of items enriched with TO field as \begin{math}
  I_x''
\end{math}.\\
\indent After parsing \begin{math}
  X
\end{math} next state will be obtained by performing GOTO over \begin{math}
  I_x
\end{math} with symbol \textit{X}. Let that state be \begin{math}
  I_y
\end{math}. Now we must have,\\
\centerline{\begin{math}
  I_y = GOTO (I_x, X)
\end{math}}. After performing GOTO we get\\
\centerline{\begin{math}
  [B \rightarrow cX.C_t\beta, j] \in I_y
\end{math}}
\centerline{\begin{math}
  [C_t \rightarrow .C_t\#_t, d] \in I_y
\end{math}}
\centerline{\begin{math}
  [C_t \rightarrow .A\gamma, e] \in I_y
\end{math}}
\centerline{\begin{math}
  [A \rightarrow .b\beta, f] \in I_y
\end{math}}
\indent After applying TO procedure over \begin{math}
  I_y
\end{math}, we get the TO fields of items \begin{math}
  [B \rightarrow cX.C_t\beta, j]
\end{math} \begin{math}
  [C_t \rightarrow .C_t\#_t, d]
\end{math} \begin{math}
  [C_t \rightarrow .A\gamma, e]
\end{math} \begin{math}
  [A \rightarrow .b\beta, f]
\end{math} as \begin{math}
  S^t
\end{math}. To enrich JSR items of LR(1) items in \begin{math}
  I_y
\end{math} we would apply FROM procedure as \begin{math}
  FROM (I_x'', I_y)
\end{math}. Now, we could see that the FROM field of JSR item for \begin{math}
  [C_t \rightarrow .A\gamma, e]
\end{math} will be equal to the FROM field of JSR item for \begin{math}
  [B \rightarrow cX.C_t\beta, j]
\end{math} which in turn is equal to the TO field of \begin{math}
  [B \rightarrow c.XC_t\beta, j]
\end{math}, which is \begin{math}
  S^h
\end{math}. So, the set of JSR items (\begin{math}
  I_y^A
\end{math}) corresponding to \begin{math}
  I_y
\end{math} contains the following items:\\
\centerline{\begin{math}
  [B \rightarrow cX.C_t\beta, j, S^h, S^t] \in I_y^A
\end{math}}
\centerline{\begin{math}
  [C_t \rightarrow .C_t\#_t, d, S^h, S^t] \in I_y^A
\end{math}}
\centerline{\begin{math}
  [C_t \rightarrow .A\gamma, e, S^h, S^t] \in I_y^A
\end{math}}
\centerline{\begin{math}
  [A \rightarrow .b\beta, f, S^h, S^t] \in I_y^A
\end{math}}
Let the state reached after performing shift of \textit{b} from the state \begin{math}
  I_y
\end{math} be \begin{math}
  I_z
\end{math}. Then, \\
\centerline{\begin{math}
  I_z = GOTO (I_y, b) = CLOSURE (\left \lbrace[A \rightarrow b.\beta, g]~|~\forall~g \in FIRST (\beta f)\right \rbrace)
\end{math}}
\indent Moreover, the TO and FROM fields of all JSR items of \begin{math}
  I_z
\end{math} will be \begin{math}
  S^t
\end{math}.\\
So, we have \begin{math}
  I_z = I_n
\end{math}. This shows that the states \textit{z} and \textit{n} are same. Let us name these states as \textit{q}. Also, the TO fields for the items of \begin{math}
  I_z
\end{math} and \begin{math}
  I_n
\end{math} are same i.e. \begin{math}
  S^t
\end{math}. Hence, in this case Shift on terminal \textit{b} in states \textit{m} and \textit{x} results only in one state \textit{q} and in the sub-table for incremental category \begin{math}
  C_t
\end{math}.\\
Theorem is \textbf{proved} in this case.\\
\textit{Case 1.2}: If \begin{math}
B \rightarrow cXC_t\beta
\end{math} is used. Let \textit{x} is the state reached before performing the shift on \textit{b}. Then, the set of items, say \begin{math}
  I_x
\end{math} related to state \textit{x} will contain following items:\\
\centerline{\begin{math}
  [B \rightarrow .C_t\beta, j] \in I_y^A
\end{math}}
\centerline{\begin{math}
  [C_t \rightarrow .C_t\#_t, d] \in I_y^A
\end{math}}
\centerline{\begin{math}
  [C_t \rightarrow .A\gamma, e] \in I_y^A
\end{math}}
\centerline{\begin{math}
  [A \rightarrow .b\beta, f] \in I_y^A
\end{math}}\\
\indent After applying TO operation to \begin{math}
  I_x
\end{math}, the TO fields of all JSR items for above LR(1) items will be \begin{math}
  S^t
\end{math}.\\
Let \begin{math}
  I_y
\end{math} be the state reached after performing shift on \textit{b} in the state \begin{math}
  I_x
\end{math}. So,
\centerline{\begin{math}
  I_y = GOTO (I_x, b) = CLOSURE (\{[A \rightarrow b.\beta, g]~|~\forall~g \in FIRST (\beta f)~\})
\end{math}}
\indent Also, the TO and FROM fields of all the JSR items for the above set of LR (1) items will be \begin{math}
  S^t
\end{math}.\\
\indent Hence, \begin{math}
  I_y = I_n
\end{math}. This shows that the states \textit{y} and \textit{n} are same. Let us name the state as \textit{q}. Moreover, the TO fields of LR(1) items of \begin{math}
  I_y
\end{math} and \begin{math}
  I_n
\end{math} are \begin{math}
  S^t
\end{math}. Hence, in this case Shift on terminal \textit{b} in the states \textit{m} and \textit{x} results only in one state \textit{q} in the subtable of \begin{math}
  C_t
\end{math}.\\
Theorem is proved in this case.\\
\\
\indent As Theorem has been proved in \textit{Case 1.1} and \textit{Case 1.2}. So, for \textbf{Case 1} also the Theorem has been proved.\\
\indent We have proved for the case containing both productions. Two other cases are when only one of these productions is present. The proof of both of these cases are very similar to the \textbf{Case 1}.\\
\indent Please note that with the given set of conditions in the Theorem, we couldn't have the case in which none of the productions belong to this set. 
\end{proof}

THEOREM 1 is crucial to BPP. In the succeeding sections we will use this theorem to create our parallel parsing algorithm.
\section{Block Parallelized Parser}
In this section we will present our Block Parallelized Parser for LR (1) grammars. We will first give the intuition and working of BPP. Then we will present our algorithm and give a proof that our parser can accept all the LR(k) and LALR (k) languages which can be accepted by a Shift Reduce LR(k) and LALR (k) parser.\\
\indent Let \begin{math}
  G' = (N', T', P', S')
\end{math} be the augmented grammar. The incremental categories are the Non Terminals associated with the blocks to be parsed in parallel. For example, if we want to parse class definitions in parallel then we can define \textit{class-definition} as the incremental category. Other examples can be \textit{function-definition}, \textit{if-statement}, \textit{for-loop}, \textit{while-loop} if they can be parsed in parallel. For most of the programming languages including C, C++, Java, C\# above blocks can be parsed in parallel. In general we can define an incremental category to be the non terminals which derive a string containing the start of the block symbol, \begin{math}
  s_b
\end{math} and ending with the end of the block symbol \begin{math}
  e_b
\end{math}. In mathematical terms, for BPP the set of incremental category \textit{IC} is defined as:
\centerline{\begin{math}
  IC = \left \lbrace C_t~|~if~C_t \rightarrow \alpha X~\in P'~then~C_t \overset{\ast}{\Rightarrow}\alpha s_b \beta e_b \right\rbrace 
\end{math}}\\
In the C Programming Language, a \textit{function-definition} can be represented by the following context free grammar productions:\\
\begin{math}
  function\mbox{-}definition \rightarrow type~name~\left( arguments \right)~block
\end{math}\\
\begin{math}
  block \rightarrow s_b~statement^\ast~e_b
\end{math}\\
\begin{math}
  statement \rightarrow type~name~;
\end{math}\\
\begin{math}
  statement \rightarrow if\mbox{-}stmt
\end{math}\\
\begin{math}
  statement \rightarrow for\mbox{-}loop
\end{math}\\
\begin{math}
  statement \rightarrow while\mbox{-}loop;
\end{math}\\
\begin{math}
  if\mbox{-}stmt \rightarrow \textbf{if}~\left(expression\right)~block
\end{math}\\
\begin{math}
  while\mbox{-}loop \rightarrow \textbf{while}~\left(expression\right)~block
\end{math}\\
\begin{math}
  for\mbox{-}loop \rightarrow \textbf{for}~\left(expression;~expression;~expression\right)~block
\end{math}\\
\indent According to the definition of Incremental Categories above, we can see that \textit{function-definition}, \textit{if-stmt}, \textit{while-loop}, \textit{for-loop} follows the condition for an incremental category.\\
\indent In a programming language there could be many types of blocks like in Java, a block related to the class definition may contain only method definitions and declaration of member variables while a block related to the method definition would contain statements including expression statements, variable declarations, if statement or loop etc. This means in a programming language not all blocks contains same type of statements. Hence, encountering the start symbol of block doesn't give us enough information about what kind of statements the block would contain. To overcome this problem we will modify the productions of incremental category such that a reduce action will happen when the start symbol is encountered. Modify the productions of \begin{math}
  C_t \in IC
\end{math} in such a way that every production \begin{math}
  C_t \rightarrow \alpha _t s_b X_t e_b
\end{math} is split into two productions:\\
\centerline{\begin{math}
  C_t \rightarrow A_t s_b X_t e_b
\end{math}}
\centerline{\begin{math}
  A_t \rightarrow \alpha _t
\end{math}}
\indent If the productions of incremental category \begin{math}
  C_t
\end{math} is structured as above then during the parsing of a word related to this incremental category there will be reduce action to reduce the current symbols to production \begin{math}
  A_t \rightarrow \alpha _t
\end{math} when \begin{math}
  s_b
\end{math} becomes the look ahead symbol. As each \begin{math}
  A_t
\end{math} is related to only one \begin{math}
  C_t
\end{math} and vice-verse, we can easily determine which incremental category has to be parsed next. \\
\indent Now we are in a stage to define Block Parallelized Grammar. Given a grammar \begin{math}
  G = (N, T, P, S)
\end{math} and a set of incremental categories \\ \centerline{\begin{math}
  IC = \left \lbrace C_t~|~C_t\in \alpha s_b X e_b~\in~P'~and~C_t \overset{\ast}{\Rightarrow} \alpha s_b \beta e_b \right \rbrace
\end{math}} we define Block Parallelized Grammar \begin{math}
  G^P = (N^P, T, P^P, S)
\end{math} such that\\
\centerline{\begin{math}
  P^P = P \cup \left \lbrace C_t \rightarrow A_t s_b X_t e_b,~A_t \rightarrow \alpha _t~|~C_t \rightarrow \alpha _t X_t \in P \right \rbrace - \end{math}} 
  \centerline{\begin{math}\left \lbrace C_t \rightarrow \alpha _t X_t~|~C_t \rightarrow \alpha _t X_t \in P \right \rbrace
\end{math}}
\centerline{\begin{math}
  N^P = \left \lbrace A_t~|~C_t \rightarrow A_t s_b X_t e_b,~A_t \rightarrow \alpha _t~\forall~C_t \rightarrow \alpha _t X_t \in P \right \rbrace  
\end{math}}
and \begin{math}
  N^P
\end{math} is partitioned using FIRSTNT as given in Section 7.\\
\indent Now we can use THEOREM 1 to create BPP. Let us have a string \begin{math}
  w = "\omega a \delta s_b \eta e_b \mu" 
\end{math} where \begin{math}
  \omega,~\delta,~\eta,~\mu \in T^\ast,~A_t \overset{\ast}{\Rightarrow} a\delta
\end{math} and \begin{math}
  X_t \overset{\ast}{\Rightarrow}s_b \eta e_b
\end{math}. During the parsing \textit{w}, when \begin{math}
  s_b
\end{math} is encountered we should have a reduce action to \begin{math}
  A_t
\end{math}, based on the production \begin{math}
  A_t \rightarrow \alpha _t
\end{math}. Now, we can get \begin{math}
  C_t
\end{math} associated with \begin{math}
  A_t
\end{math}. According to THEOREM 1, during parsing of the word \textit{w} if the parser reaches at state \textit{q} in the sub table of \begin{math}
  C_t
\end{math} after performing shift action on \textit{a} then during the incremental parsing of the word \begin{math}
  w_t = "a\delta s_b \eta e_b"
\end{math} for the incremental category \begin{math}
  C_t
\end{math} the parser will also reach the state \textit{q} in the sub table of \begin{math}
  C_t
\end{math} after performing the shift action on \textit{a}. This means, we can replace the state reached just before performing shift action on \textit{a} with the start state of subtable of \begin{math}
  C_t
\end{math} and \begin{math}
  w_t
\end{math} can now be parsed incrementally. \\
\indent It is now evident that why the partitions of Non Terminals
should be created as described in Section 7. If the partitions are not created in this way, then it may happen after a shift on \textit{a} during the parsing of \textit{w} and incremental parsing of \begin{math}
  w_t
\end{math} may reach the same state but not in the same sub-table. By creating partitions as described in Section 8, we make sure that
when \begin{math}
  s_b
\end{math} is encountered by BPP then the newly created parallel parser knows in which sub-table it has to continue the parsing in. On the other hand if partitions are not created as described in Section 8, then newly created parallel parser wouldn’t know in which sub-table it has to continue parsing in.\\
\indent This property is used by BPP to parse the block incrementally. Algorithm 1 is the BPP parsing algorithm of incremental table \begin{math}
  S^t
\end{math}. If \textit{t} = 1, then the algorithm is for the Top Level Block. Otherwise it is for other incremental category.\\
\begin{algorithm}
\caption{Block Parallelized Parsing Algorithm}
\begin{algorithmic}[1]
\State $a = start~symbol~of~input$
\State $\$ = last~symbol~of~input$
\State $h = initial~parsing~subtable$
\State $stack = stack~of~states$
\While {$true$}
\State $s = stack.top ()$
\If {$Tab (S^h) [s, a] == shift~t$}
\State $stack.push (t)$
\State $a = next~symbol~of~input$ 
\ElsIf {$Tab (S^h)[s, a] == jump~k$}
\State $h = k$
\ElsIf {$Tab (S^h)[s, a] == reduce~A\rightarrow \beta$}
\State $b = a$
\If {$a == s_b~and~Tab(S^h)[s, a] == reduce~A_t~\rightarrow \alpha _t$}
\State $get~C_t~related~to~A_t$
\State $stack_t = new~stack~of~states$
\State $stack_t.push (start~state~of~C_t)$
\State $pop~|\alpha _t|~states~from~stack_t~and~push~them~to~stack_t$
\State $create~new~block~parser~related~to~C_t~with~stack_t$
\State $start~new~block~parser$
\State $go~to~the~end~of~this~block$
\State $a = C_t$
\Else
\State $pop~|\beta|~states~from~stack$
\State $t = stack.top ()$
\If {$Tab (S^h)[t, A] == shift~p$}
\State $stack.push (p)$
\EndIf
\EndIf
\ElsIf {$Tab (S^h)[s, a] == accept$} \Return
\Else 
\State $error$
\EndIf
\EndWhile
\end{algorithmic}
\end{algorithm}
This parsing algorithm is for any block be it top level block or child block. Lines 1-4 initializes different local variables. Lines 5-32 is the main loop of algorithm which does the whole work. Line 6, gets the top state on the stack. Lines 7-9 pushes the next state on the stack if there is a shift operation. Similarly, lines 10-11 changes current table if there is a jump operation. Line 12-27 are executed if there is reduce action which reduces according to the production \begin{math}
  A \rightarrow \beta
\end{math}. Line 14 checks if current input symbol is a start of the block symbol and if reduce action reduces \begin{math}
  A_t \rightarrow \alpha _t
\end{math} for an incremental category \begin{math}
  C_t
\end{math}. If yes then Lines 15-22 gets \begin{math}
  C_t
\end{math} related to \begin{math}
  A_t
\end{math}, pops \begin{math}
  |\alpha _t|
\end{math} states from stack and pushes these states and start state to a new stack, creates and starts a new BPP for \begin{math}
  C_t
\end{math} and shifts to the end of block. In this case next symbol will become \begin{math}
  C_t
\end{math}. If check of Line 14 fails then it means this is a regular reduce action not associated with any block. Lines 24-27, pops \begin{math}
  |\beta|
\end{math} states from stack and shifts to a new state. Line 28 returns if there is an accept action. Accept action can be for both Top Level block and Child Block. Line 30 reports error if none of the above cases are satisfied. \\
\indent Fig. 2 shows an example of how BPP works. It will start parsing the block of function \textit{f}. When it will encounter an \textit{if} block a new BPP in another thread will be created which will parse \textit{if} block. Parent BPP will move ahead to the end of \textit{if} block and will also create another thread to parse \textit{else} block. In this way input is divided into different threads parsing each block. \\
\begin{figure}[t!]
\includegraphics[width=85mm]{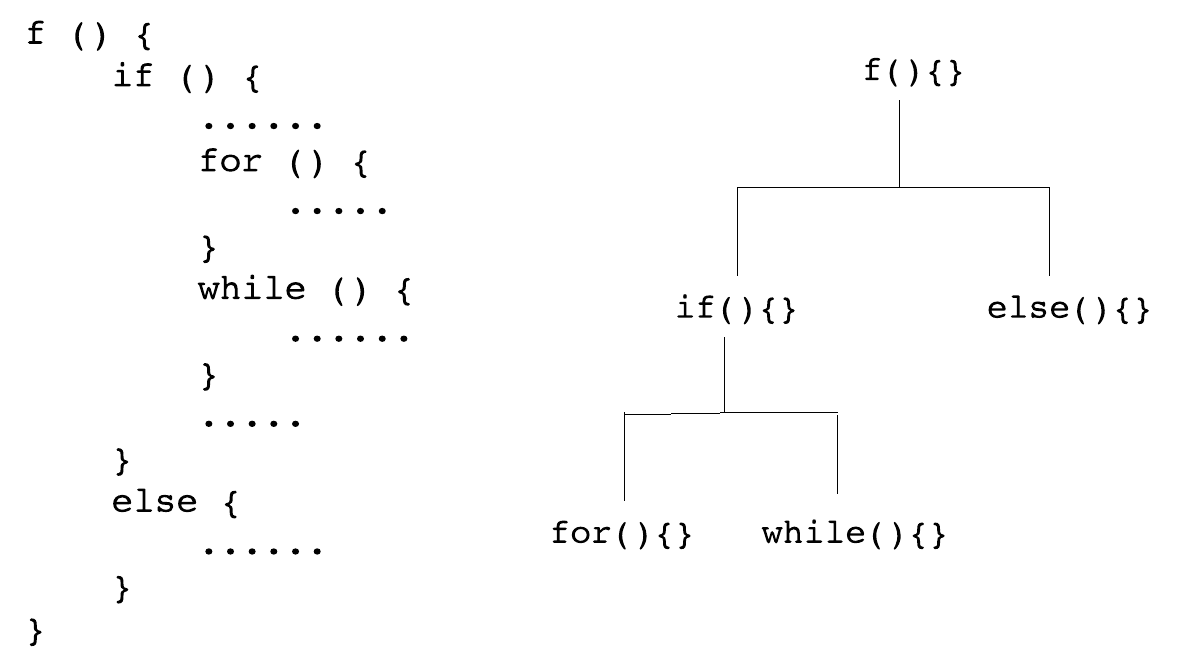}
\caption{Example of BPP parsing source code}
\end{figure}
\indent In this algorithm we have tried to minimize the
amount of serial work to be done to get to the end of the block for
the new block parser. One BPP doesn’t have to do any communication with other BPPs. Also, there are no side effects of the above algorithm. All the variables which are being modified are local variables. Hence, there is no need of synchronization. This also reduces any amount of cache contention between different processors. Generation of Abstract Syntax Tree or Parsing Tree is easy using above algorithm and it requires very little change in the above algorithm. \\
\indent It may be argued that the step “go to the end of this block” is a serial bottleneck for the parallel algorithm. \cite{hillis} describes an algorithm to perform lexical analysis of string in \textit{O(log n)} time using \textit {O(n)} processors in a parallel fashion. When performing lexical analysis in parallel as described in \cite{hillis}, lexer could store the start symbols of a block with its corresponding end symbol for that block. Now determining the end of block is just a matter of searching through the data structure. Many ways exist to make this searching as fast as possible like using a Binary Search Tree or a Hash Table.
\section{Comparison with other Parallel Parsing Algorithms}
In this section we will show how our technique is better than other
techniques. \cite{khanna} developed a technique which divides whole
grammar into n sub- grammars which are individually handled by
n sub-compilers. Each sub-compiler needs its own scanner which
can scan a sub-grammar. It requires an automatic tool to generate
sub-compiler. This technique requires significant changes in not
only in Parser and Grammar but also in Lexical Analyzer phase.
Contrast to this our Block Parallelized Parser is easy to generate
as our technique does not require any change in the grammar and
lexical analyzer and it is pretty easy to modify current YACC and
Bison Parser Generator tools to support the generation of our
parser.\\
\indent LR substring parsing technique described in \cite{clarke} is specifically
for Bounded Context – (1, 1) grammars. There are no limitations
like this to Block Parallelized Parser. Although in this paper we
have shown how we can create an LR (1) Block Parallelized
Parser but we believe it can be extended to LR (k) class of
languages and also could be used by LALR (1) parser. Hence, our
technique accepts a larger class of grammars.\\
\indent \cite{cohen}, \cite{schell}, \cite{fischer}, \cite{ligett} all develops algorithm for parsing LR (k)
class of languages in parallel. These techniques and in all other
techniques the creation of Abstract Syntax Tree is not as easy as itis in our technique. Moreover our technique is simpler than all
others.\\
\indent Hence, we could see that Block Parallelized Parser is easy to
construct, accepts wider class of languages and supports an easy
construction of Abstract Syntax Tree.
\section{Implementation and Evaluation}
We implemented Lexer, Block Parallelized Parser and Shift Reduce LR (1) parser for C Programming Language supporting a few GNU C extensions required for our tests. Implementation was done in C\# Programming Language. To simplify our implementation we only included function-definition as the Incremental Category for BPP. Moreover, function-definition would still give us a sufficient amount of parallelism as we would see in the evaluation. We modified the Lexer phase so that it will
keep track of the position of s b and its corresponding e b . This information was stored in the form of a C\# Dictionary (which is implemented as a Hash Table) with the position of s b as the key and position of e b as the value. As, thread creation has significant overhead so we used C\# TaskParallel Library which is Thread Pool implementation in C\#. Our implementation doesn’t have a C preprocessor implementation. So, we first used gcc to perform preprocessing and the preprocessed file is used as input to our implementation.\\
\indent We evaluated the performance of BPP with Shift Reduce LR
(1) parser by parsing 10 random files from the Linux Kernel source code. We compiled C\# implementation using Mono C\# Compiler 3.12.1 and executed the implementation using Mono JIT Compiler 3.12.1 on machine running Fedora 21 with Linux Kernel 3.19.3 with 6 GB RAM and Intel Core i7-3610 CPU with HyperThreading enabled.\\
\indent In C Programming Language, preprocessing \texttt{\#include} could
actually generate very long files. Normally, the header files contains declarations not function definitions. So, this leads to less amount of parallelism being available. Hence we decided to evaluate with including header files and excluding header files. Fig. 3 shows the performance improvement with respect to Shift Reduce LR(1) parser of 10 random Linux Kernel files. Fig. 3 shows performance improvement for both cases including the header files and excluding the header files. As expected we could see that the performance improvement with excluding the header files is more than the performance improvement including the header files.\\
\begin{figure}[t!]
\includegraphics[width=85mm]{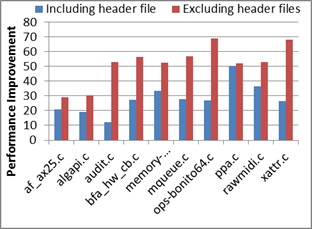}
\caption{Results of parsing 10 Linux Kernel Source Files.}
\end{figure}
\indent The performance improvement in the case of excluding header
files matters most for the programming languages like Java,
Python, C\# where most of the program is organized into blocks
the results because in these programs the amount of parallelism
available is high.\\
\indent The average performance improvement in the case of excluding
header files is 52\% and including header files is 28\%.
\section{Conclusion}
In this document we present our Block Parallelized Parser
technique which could parse the source code in a parallel fashion.
Our approach is a divide and conquer approach in which, the
divide step divides the source code into different blocks and parse
them in parallel whereas the conquer step only waits for all the
parallel parsers to complete their operation. It is based on the
Incremental Jump Shift Reduce Parser technique developed by
\cite{degano}. Our technique doesn’t require any communication in
between different threads and doesn’t modify any global data.
Hence, this technique is free of thread synchronization. We
develop this technique for LR (1) languages and we believe that it
can be extended to accept LR (k) languages and could be
converted to an LALR (1) parser easily. Our technique doesn’t do
any major changes in the parsing algorithm of a Shift Reduce
Parser hence the Abstract Syntax Tree can be created in the same
way as it has been creating in Shift Reduce Parser. Moreover, our
parser can also work as an Incremental Block Parallelized Parser.
We implemented Block Parallelized Parser and Shift Reduce LR
(1) Parser for C Programming Language in C\#. The performance
evaluation of BPP with Shift Reduce LR (1) parser was done by
parsing 10 random files from the Linux Kernel source code. We
compiled C\# implementation using Mono C\# Compiler 3.12.1 and
executed the implementation using Mono JIT Compiler 3.12.1 on
machine running Fedora 21 with Linux Kernel 3.19.3 with 6 GB
RAM and Intel Core i7-3610 CPU with HyperThreading enabled.
We found out that our technique showed 28\% performance
improvement in the case of including header files and 52\%
performance improvement in the case of excluding header files.
\section{Future Work}
Our parser accepts LR (1) languages we would like to extend it to
accept LR (k) languages. In our technique, the parser determines
when to create a new parallel parser thread. If the responsibility of
this decision can be given to the lexical analysis phase then the
lexical analysis can actually start the parsers in parallel. This will
lead to significant performance advantage. Moreover, our
technique has been applied to languages which doesn’t have
indentation in their syntax like the way Python has. \cite{adams} shows an
efficient way to parse the language which has indentation as a
mean to determine blocks. Our parser can be extended to accept
those languages also. We are working towards developing a Block
Parallelized Compiler which could compile different blocks of a
language in parallel. Block Parallelized Parser is one of the
components of a Block Parallelized Compiler. Semantic Analysis
phase also share the same properties as the Syntax Analysis phase.
In Programming Languages, an entity like variable or type is
declared before using it. So, in this case also a lot can be done to
actually parallelize the semantic analysis phase.

\end{document}